\documentclass[11pt]{article}

\usepackage{fullpage}
\usepackage{amsthm, latexsym,amssymb,amsfonts,amsmath,mathrsfs,float}
\usepackage[font=small,labelfont=bf, width=.618\textwidth]{caption} 
\usepackage[dvipsnames]{xcolor}
\usepackage[T1]{fontenc}

\newcommand{\Lqn}{\mathcal L(N;{\bf p})}
 
\usepackage[linktocpage=true, linkbordercolor=red]{hyperref}
\hypersetup{
     colorlinks   = true,
     citecolor    = black,
     urlcolor = WildStrawberry
} 

  \newcommand{\R}{\mathbf{R}} 
  \newcommand{\N}{\mathbb{N}} 
  \newcommand{\pmset}[1]{\{-1,1\}^{#1}} 
  \newcommand{\bset}[1]{\{0,1\}^{#1}} 
  
  \newcommand{\st}{:\,} 

  \newcommand{\Exp}{{\mathbb{E}}}

  \DeclareMathOperator{\sign}{sign}
  \newcommand{\poly}{\mbox{\rm poly}}

  \renewcommand{\Pr}{\mbox{\rm Pr}}
  
  \newcommand{\eps}{\varepsilon}

\newcommand{\beq}{\begin{equation}}
\newcommand{\eeq}{\end{equation}}
\newcommand{\beqn}{\begin{equation*}}
\newcommand{\eeqn}{\end{equation*}}
\newcommand{\beqr}{\begin{eqnarray}}
\newcommand{\eeqr}{\end{eqnarray}}
\newcommand{\beqrn}{\begin{eqnarray*}}
\newcommand{\eeqrn}{\end{eqnarray*}}
\newcommand{\bmline}{\begin{multline}}
\newcommand{\emline}{\end{multline}}
\newcommand{\bmlinen}{\begin{multline*}}
\newcommand{\emlinen}{\end{multline*}}

\def\thmstyle{\it} 
\def\@begintheorem#1#2{\it \trivlist \item[\hskip
        \labelsep{\bf #1\ #2.}]\thmstyle}
\def\@opargbegintheorem#1#2#3{\it \trivlist \item[\hskip
        \labelsep{\bf #1\ #2\ (#3).}]\thmstyle}

\newtheorem{theorem}{Theorem}[section]
\newtheorem{proposition}[theorem]{Proposition}
\newtheorem{lemma}[theorem]{Lemma}

\newtheorem{corollary}[theorem]{Corollary}
\newtheorem{definition}[theorem]{Definition}

\newtheorem{remark}{Remark}

\makeatletter
\newtheorem{rep@theorem}{\rep@title}
\newcommand{\newreptheorem}[2]{%
\newenvironment{rep#1}[1]{%
 \def\rep@title{#2 \ref{##1}}%
 \begin{rep@theorem}}%
 {\end{rep@theorem}}}
\makeatother

\newreptheorem{lemma}{Lemma}
\newreptheorem{theorem}{Theorem}
\newreptheorem{corollary}{Corollary}
\newreptheorem{proposition}{Proposition}

\begin{document}

\title{
\huge
On Embeddings of $\ell_1^k$ from
Locally Decodable Codes
}
\author{
Jop Bri\"et\thanks{Center for Mathematics \& Computer Science (CWI), The Netherlands.
Funded by a Rubicon grant from the Netherlands Organisation for Scientific Research (NWO). E-mail: \texttt{j.briet@cwi.nl}}}
\date{}
\maketitle


\begin{abstract}
We show that any $q$-query locally decodable code (LDC) gives a copy of~$\ell_1^k$ with small distortion in the Banach space of $q$-linear forms on~$\ell_{p_1}^N\times\cdots\times\ell_{p_q}^N$, provided~$1/p_1 + \cdots + 1/p_q \leq 1$ and where~$k$, $N$, and the distortion are simple functions of the code parameters.
We exhibit the copy of~$\ell_1^k$ by constructing a basis for it directly from ``smooth'' LDC decoders.
Based on this, we give alternative proofs for known lower bounds on the length of 2-query LDCs.
Using similar techniques, we reprove known lower bounds for larger~$q$.
We also discuss the relation with an alternative proof, due to Pisier, of a result of Naor, Regev, and the author on cotype properties of projective tensor products of~$\ell_p$ spaces.
\end{abstract}

%
\section{Introduction}


\paragraph{Locally decodable codes.}
A locally decodable code (LDC) is an error correcting code that maps  a message string into a codeword such that, even if part of the codeword is changed adversarially, any single message symbol can be retrieved by querying only a small number of randomly selected codeword coordinates.
More formally, for positive integers~$k$, $N$, and~$q$, real numbers~$\delta,\eps\in (0,1/2]$, and a finite alphabet~$\Gamma$, a map $C: \bset{k}\to\Gamma^N$ is a $(q,\delta,\eps)$-\emph{locally decodable code} if there exists a \emph{decoder} (a probabilistic algorithm)~$\mathcal{A}$  such that:
 
\begin{itemize}
\item For every message $x\in\bset{k}$, index $i\in[k]$, and string~$y\in\Gamma^{N}$ that differs from the codeword~$C(x)$ in at most~$\delta N$ coordinates,
\beqn
\Pr[\mathcal A(i, y) = x_i] \geq \frac{1}{2} + \eps.
\eeqn
\item $\mathcal{A}$ (non-adaptively) queries at most $q$ coordinates of $y$. 
\end{itemize}
The most general decoder first samples a set ${S\subseteq[N]}$ of at most~$q$ codeword coordinates from a probability distribution that depends on~$i$ only.
Then, it outputs a random bit whose distribution depends only on~$i$, $S$, and the sequence~$(y_s)_{s\in S}$ of (possibly corrupted) codeword entries at~$S$.\footnote{Adaptive decoders, whose queries depend on the values of previously queried coordinates, can be made non-adaptive at the cost of a factor $1/|\Gamma|^{q-1}$ in the decoding bias~$\eps$.}

The central problem regarding LDCs is to determine the smallest possible codeword length~$N$ as a function of the message length~$k$ for various ranges of the query complexity~$q$ and alphabet size~$|\Gamma|$ when~$\delta$ and~$\eps$ are fixed constants.
\medskip

\paragraph{Synopsis.}
The main result of this paper (Theorem~\ref{thm:main} below) connects LDCs to a geometric property of certain finite-dimensional normed vector spaces. 
In particular, we show explicitly that LDCs give  linear low-distortion embeddings of $\ell_1^k$ in such spaces (for $k$ as above).
This link follows implicitly from work of Pisier~\cite{Pisier:1973} and Naor, Regev, and the author~\cite{Briet:2012d},
and an explicit instance of it based on specific LDCs was first shown by Pisier~\cite{Pisier:bedlewo}\footnote{Unfortunately no proceedings for this workshop appear to be published.} en route to an alternative proof of the main result of~\cite{Briet:2012d}.
Theorem~\ref{thm:main} applies to any LDC and our proof is slightly more direct than Pisier's argument.
Since he merely used what was sufficient for his purpose, much of the content of this paper may well have been known to him at the time, and for this reason this paper may be regarded as partly expository.
 

The main complexity-theoretic message here is that
thanks to a known upper bound on the dimension~$k$ for which certain normed spaces of bilinear forms (or matrices) can accommodate~$\ell_1^k$, we obtain new proofs for known lower bounds on the length of 2-query LDCs with binary alphabets and alphabets of non-constant size. 
More generally, the above-mentioned link suggests a new avenue to explore for proving such bounds when~$q\geq 3$,  for which techniques are currently in short supply.
In similar geometric spirit, but inspired by techniques used by Kerenidis and de~Wolf~\cite{Kerenidis:2004}, we also reprove known lower bounds for LDCs with a larger number of queries.


\paragraph{Origins and applications.}
The notion of LDCs originated from works on probabilistically checkable proofs~\cite{Babai:1991a, Sudan:1992} and private information retrieval (PIR)~\cite{Chor:1998}, though they where first formally defined by~Katz and Trevisan~\cite{Katz:2000} in the context of noisy data transmission.
Since then, the range of areas where these codes turn out to play a role has grown steadily. Applications in theoretical computer science now include polynomial identity testing~\cite{Dvir:2007}, data structures~\cite{deWolf:2009, Chen:2013}, and complexity theory~\cite{Dvir:2010}. 
In pure mathematics, they recently found applications in discrete geometry \cite{Barak:2011, Briet:2014c} and Banach spaces~\cite{Briet:2012d}.

\paragraph{Constructions.}
Four constructions currently roughly cover the best-known trade-offs between codeword length, query complexity, and alphabet size.
The family of Reed-Muller codes, which work based on polynomial interpolation, give LDCs for a large range of parameters~\cite{Babai:1991a}.
For example, the Hadamard code is a binary 2-query LDC of length~$N = 2^k$ and for constant-sized alphabets the Reed-Muller family gives LDCs with query complexity $q=\poly(\log k)$ and length $\poly(k)$.
Great strides were made recently with the discovery of Matching Vector codes (MV-codes)~\cite{Yekhanin:2007, Efremenko:2009, Dvir:2010a} and Multiplicity codes~\cite{Kopparty:2011}, which outperform Reed-Muller codes in constant and $\poly(k)$ query-complexity regimes, respectively.
See~\cite{Yekhanin:2012} for a detailed survey, and for recent work on high query complexity expander-based codes, see~\cite{Hemenway:2014, Kopparty:2015}.
Since our focus will be on the constant query complexity regime, we highlight that for constant~$q\geq 3$ there are $q$-query MV-codes with length $\exp\exp(o(\log k))$ and constant alphabet size.

The best constructions of  LDCs over large alphabets come from PIR schemes.\footnote{With information-theoretic security.}
A PIR scheme replicates a $k$-bit database among~$q\geq 2$ non-communicating servers that interact with a user wishing to know some entry $i\in[k]$ of the database that he/she wants to keep  hidden from the servers.
The goal is to find a scheme that achieves the above with minimal communication between the user and the servers.
Katz and Trevisan~\cite{Katz:2000} observed (and~\cite{Goldreich:2006} showed formally) that $q$-query LDCs are essentially equivalent to $q$-server PIRs where communication proceeds in two rounds and the total number of communicated bits per index~$i\in[k]$ is given by~$2\log (|\Gamma|N)$.
A recent breakthrough of Dvir and Gopi~\cite{Dvir:2014} gave two-round $q$-server PIRs with  communication cost $\exp(o(\log k))$; these schemes rely on the same combinatorial objects, called ``matching vector families,'' as MV-codes.
Most remarkably, their construction shows that there exist 2-query LDCs whose alphabet size \emph{and} length is~$\exp\exp(o(\log k))$. 

\paragraph{Lower bounds.}
What we know about the \emph{necessary} length of LDCs has changed little during the last decade
and most of the best-known lower bounds are far from the parameters of the best-known constructions.
There are currently two general 
cases where optimal bounds are known.
First, it was shown in~\cite{Katz:2000} that independent of the code length, 1-query LDCs can only encode a constant number of message bits once we fix~$\delta$, $\eps$, and~$|\Gamma|$.
The second case concerns binary 2-query LDCs. Those turn out to require exponential length, as is achieved by the Hadamard code.
The original proof of the exponential bound due to Kerenidis and de~Wolf~\cite{Kerenidis:2004}, which is based on quantum-information-theoretic arguments, gives the bound 
\beqn\label{eq:kdw}
N \geq 2^{\Omega(\delta\eps^2 k)}.
\eeqn
Ben-Aroya, Regev and de~Wolf~\cite{Ben-Aroya:2008} obtained a similar bound using a Fourier-analytic inequality for matrix-valued functions, which they derived from a deep result from Banach space theory on uniform convexity of Schatten-1~\cite{Ball:1994}.
These proofs also form the basis for the best-known lower bounds for~$q \geq 3$.
For even integers~$q\geq 4$ and constant~$\delta$ and $\eps$, \cite{Kerenidis:2004} used a reduction to maps akin to 2-query LDCs to prove that binary $q$-query LDCs have length 
$
\Omega((k/\log k)^{q/(q-2)}).
$
A similar reduction gives the same bound based on~\cite{Ben-Aroya:2008}.
Later, Woodruff~\cite{Woodruff:2007a} slightly improved this bound to~$\Omega(k^{q/(q-2)}/\log k)$ using a more careful reduction.
Oddly, for odd ${q\geq 3}$, we do not know how to prove better lower bounds other than by using the ones for~$q+1$ queries.

For 2-query LDCs over large alphabets, Wehner and de~Wolf \cite{Wehner:2005} proved the lower bound ${|\Gamma|^2\log N \geq \Omega(k)}$, which implies an~$\Omega(\log k)$  bound on the communication required in any (two-round) 2-server PIR scheme.\footnote{The current best constant is obtained by combining their result with the bound $\log N \geq 2\log k - 2\log|\Gamma| - O(1)$ due to~\cite{Katz:2000}, which gives
$(5 - o(1))\log k$.} 
Their proof also used quantum information theory.
Here too, there thus remains  a large gap with  the best construction. 
Slightly better bounds are known if the alphabet~$\Gamma$ is~$\bset{n}$ and the decoder, after sampling a set of codeword coordinates, returns a random bit whose distribution depends on at most $m \leq n$ predetermined bits at those coordinates. 
Such codes may be seen as PIR schemes where the servers send the user an~$n$-bit string of which the user only reads at most~$m$ bits.
This happens, for example, in~\cite{Dvir:2014}, where the best-known constructions of matching vectors~\cite{Grolmusz:2000} give~$m\approx \sqrt{n}$.
In~\cite{Wehner:2005} it is proved that in this case,
\beqn\label{eq:wdw}
2^m\sum_{l = 0}^m{n\choose  l} \log N\geq \Omega(k).
\eeqn
For example, if ${m = n^\eta}$ for some constant~$\eta\in (0,1)$, this implies a bound of $\Omega((\log k)^{1/\eta - o(1)})$ on the communication for two-round two-server PIRs.


%
\paragraph{Banach space geometry.}
Different results from Banach space theory were used on several occasions to prove lower bounds on LDCs or similar objects~\cite{Ben-Aroya:2008, Dvir:2014b, Briet:2014c}.
In the opposite direction, the aforementioned 3-query MV-codes were used in~\cite{Briet:2012d} to \emph{solve} an open problem on Banach spaces.
The following basic definitions and facts will allow us to elaborate.
For $p\in[1,\infty]$, a distortion parameter $K\geq 1$, and a positive integer~$d$, a Banach space~$X$ is said to contain a $K$-isomorphic copy of~$\ell_p^d$ if there exist~$A_1,\dots,A_d \in X$ such for any vector~${\alpha\in\R^d}$,
\beqn
\|\alpha\|_{\ell_p} \leq \Big\|\sum_{i=1}^d \alpha_i A_i\Big\|_X \leq K\|\alpha\|_{\ell_p}.
\eeqn
The containment of copies of certain finite-dimensional~$\ell_p$ spaces is strongly linked with the notions of (Rademacher) type and cotype, which are defined as follows.
The space~$X$ has \emph{type} $p>0$ if there exists a constant~$T <\infty$ such that for every positive integer~$d$ and~$A_1,\dots,A_d\in X$, we have
\beq\label{eq:type}
\Exp_{x \in\pmset{d}}\Big\|\sum_{i=1}^d x_i A_i\Big\|_X \leq 
T\, \Big(\sum_{i=1}^d \|A_i\|_X^p\Big)^{1/p}.
\eeq
Observe that the right-hand side of~\eqref{eq:type} decreases as~$p$ increases and that by the triangle inequality, any space has type~1.
We say that a space \emph{fails nontrivial type} if there is no~$p>1$ for which it has type~$p$.
The infimum over~$T$ satisfying~\eqref{eq:type} for any~$d\in\N$ and~$A_1,\dots,A_d\in X$ is denoted by~$T_p(X)$.

A space~$X$ has \emph{cotype}~$r > 0$ if there exists a constant~$C < \infty$ such that for every positive integer~$d$ and~$A_1,\dots, A_d\in X$, we have
\beq\label{eq:cotype}
\Exp_{x \in\pmset{d}} \Big\|\sum_{i=1}^d x_i A_i\Big\|_X \geq \frac{1}{C}\, \Big(\sum_{i=1}^d \|A_i\|_X^r\Big)^{1/r}.
\eeq
By convexity of norms and Jensen's inequality, any space has cotype~$\infty$ and we say that a space \emph{fails finite cotype} if there is no~$r<\infty$ such that it has cotype~$r$.
The infimum over~$C$ satisfying~\eqref{eq:cotype} for any~$d\in\N$ and~${A_1,\dots,A_d\in X}$ is denoted by~$C_r(X)$.

As a well-behaved example, Hilbert space has type~2 and cotype~2.
Two important examples that fail one or the other are~$\ell_1$, which fails nontrivial type, and~$\ell_\infty$, which fails finite cotype; both failures are easily seen by setting the~$A_i$ to be distinct standard basis vectors.
It turns out that these are not just some examples that fail either nontrivial type or cotype, but in the sense alluded to earlier, they are the only examples.
Indeed, Pisier~\cite{Pisier:1973} showed that an infinite-dimensional Banach space $X$ fails nontrivial type if and only if there exists a~$K<\infty$ such that~$X$ contains a {$K$-isomorphic} copy of~$\ell_1^d$ for every positive integer~$d$. 
Complementing this, Maurey and Pisier~\cite{Maurey:1973} equated failure of finite cotype with containment of a $K$-isomorphic copy of~$\ell_\infty^d$ for every~$d$.

\paragraph{LDCs and copies of~$\ell_p^d$.}
The following Banach spaces are relevant to~LDCs.
For positive integers~$N$ and $q\geq 2$, and a vector ${\bf p} = (p_1,\dots,p_q)\in (1,\infty)^q$ such that $1/p_1 + \cdots + 1/p_q \leq 1$, we shall consider the real~$N^q$-dimensional vector space of $q$-linear forms on~$\R^N$
endowed with the norm
\beqn
\|A\|_{\bf p} = \sup\Big\{\frac{ A(z[1],\dots,z[q])}{\|z[1]\|_{\ell_{p_1}}\cdots \|z[q]\|_{\ell_{p_q}}}
\st
z[1],\dots,z[q]\in \R^N\smallsetminus\{{\bf 0}\}\Big\}.
\eeqn
We denote this Banach space by~$\Lqn$.
Note that~$\mathcal L(N; (2,2))$ can be identified with the space of matrices endowed with the Schatten-$\infty$ norm
and
that the spaces~$(\mathcal L(N;{\bf p}))_{N\in\N}$ are subspaces of the Banach space of bounded $q$-linear forms on~$\ell_{p_1}\times\cdots\times\ell_{p_q}$, which we denote by~$\mathcal B({\bf p})$.

In~\cite{Briet:2012d} it is shown that for fixed~$q$, $\delta$,~$\eps$, and vector~${\bf p}$ as above, the existence of an infinite family of binary $q$-query LDCs with sub-exponential length implies that for
any $r\in [2,\infty)$, the cotype-$r$ constant of the \emph{dual} of~$\mathcal L(N;{\bf p})$ satisfies
\beq\label{eq:cotype-3N}
\lim_{N\to\infty} C_r\big(\mathcal L(N;{\bf p})^*\big) = \infty.
\eeq
Since the MV-codes of~\cite{Efremenko:2009} have sub-exponential length, 
the above holds for~$q\geq 3$.
It follows that the infinite-dimensional space~$\mathcal B({\bf p})^*$ fails finite cotype, which allowed~\cite{Briet:2012d} to answer in the negative a question of~\cite{Diestel:2003} on the permanence of finite cotype under the projective tensor product.\footnote{
The space~$\mathcal B({\bf p})^*$ is precisely  the projective tensor product of the spaces~$\ell_{p_1},\dots,\ell_{p_q}$, denoted~$\ell_{p_1}\widehat\otimes\cdots\widehat\otimes\ell_{p_q}$~\cite[Chapter~2, Section~2.2]{Ryan:2002}.
The result is stated only for the case~$q = 3$ but the same proof works when~$q\geq 3$.
}
This in turn has implications for the space~$\mathcal B({\bf p})$ itself.
For any Banach space~$X$ and any~$p,r\in (1,\infty)$ such that~$1/p + 1/r = 1$, it holds that $T_p(X) \geq C_r(X^*)$ \cite[Proposition~3.2]{Pisier:1999}. 
It thus follows from~\eqref{eq:cotype-3N} that for any~$p>1$, the type-$p$ constants of $\mathcal L(N;{\bf p})$ are unbounded and hence~$\mathcal B({\bf p})$ fails nontrivial type. 
The LDCs therefore imply that there exists a $K<\infty$ such that for every~$d\in\N$, the space~$\mathcal B({\bf p})$ contains a $K$-isomorphic copy of~$\ell_1^d$.

%
\subsection{Main result}

Given that LDCs imply the \emph{existence} of copies of~$\ell_1^d$ in $\mathcal B({\bf p})$, it is natural to ask what these copies look like.
Here we give explicit constructions of those copies. 
After stating the main theorem we shall elaborate on its implications for LDC lower bounds and Banach space geometry.

\begin{theorem}\label{thm:main}
Let~$k$, $N$, and $q\geq 2$ be positive integers,~$\delta,\eps\in (0,1/2]$, and let~$\Gamma$ be a finite set.
Suppose there exists a $(q,\delta,\eps)$-LDC from $\bset{k}$ to~$\Gamma^N$.
Then, 
for any
${\bf p} \in (1,\infty)^q$ such that $1/p_1 + \cdots + 1/p_q \leq 1$,
for every integer $
N' \geq 2|\Gamma|N
$,
and for any real number
{$
K \geq 2^q|\Gamma|^{(q+2)/2}/(\delta\eps)
$},
the space $\mathcal L(N';{\bf p})$ contains a $K$-isomorphic copy of~$\ell_1^k$.
\medskip

That is, there exist~$q$-linear forms $A_1,\dots,A_k$ on~$\ell_{p_1}^{N'}\times\cdots\times\ell_{p_q}^{N'}$ (that we give explicitly) such that for any vector $\alpha\in\R^k$,
\beq\label{eq:l1copy}
\|\alpha\|_{\ell_1} \leq \Big\|\sum_{i=1}^k \alpha_i A_i\Big\|_{\bf p} \leq K\|\alpha\|_{\ell_1}.
\eeq

Moreover, if for positive integers~$m \leq n$, we have~$\Gamma = \bset{n}$ and the LDC decoder's output depends on at most~$m$ predetermined bits of each queried codeword symbol, then the above holds for
\beqn\label{eq:Kdef}
N' \geq {n\choose\leq m}N
\quad\quad
\text{and}
\quad\quad
K \geq q{n\choose \leq m}^{(q+2)/2}/\delta\eps,
\eeqn
where ${n\choose \leq m} = {n\choose 0} + {n\choose 1} + \cdots + {n\choose m}$.
\end{theorem}


\paragraph{Application for LDC lower bounds.}
If the known LDC lower bounds leave any room for improvement, then the near lack thereof in the last decade could indicate that new techniques are needed to make progress. 
Alternative techniques with which the known bounds can be reproved were already asked for by Trevisan~\cite[Question 3]{Trevisan:2004}.
Theorem~\ref{thm:main} gives a method based on showing that $\Lqn$ contains no copies of~$\ell_1^d$ for large dimension~$d$ and small distortion.
As we show in Section~\ref{sec:ldcs}, via this method we immediately recover the above-mentioned lower bounds for 2-query LDCs (up-to slightly poorer dependence on~$\delta$ and $|\Gamma|$). 
Indeed, previous results easily imply that any $O(1)$-isomorphic copy of~$\ell_1^d$ in $\mathcal L(N;(2,2))$ must satisfy~$d \leq O(\log N)$.

\paragraph{Application for cotype.}
As observed by Pisier~\cite{Pisier:bedlewo}, Theorem~\ref{thm:main} also gives an alternative route from LDCs to the result~\eqref{eq:cotype-3N} of~\cite{Briet:2012d}.
Indeed, for $q\geq 3$, the theorem combined with the parameters of $q$-query MV-codes of~\cite{Efremenko:2009} implies that~$\mathcal L(N;{\bf p})$ contains an $O(1)$-isomorphic copy of~$\ell_1^d$ for $d \geq (\log N)^{\omega(1)}$---in stark contrast with the case~$\mathcal L(N;(2,2))$ mentioned above.
If we now let the vector~$\alpha$ in Theorem~\ref{thm:main} be random and uniformly distributed over~$\pmset{k}$, then averaging~\eqref{eq:l1copy} gives that for any $p>1$, we have 
\beq\label{eq:type-qN}
T_p(\mathcal L(N;{\bf p})) \geq (\log N)^{\omega(1)}.
\eeq
A celebrated result of Pisier~\cite{Pisier:1980b} (which bounds the $K$-convexity constant of finite-dimensional Banach spaces; see also~\cite[Lemma~7, Theorem~13]{Maurey:2003}) implies that there exists an absolute constant~$c\in (0,\infty)$ such that for any finite-dimensional Banach space~$X$ and any $p,r\in (1,\infty)$ satisfying $1/p + 1/r = 1$, we have
\beq\label{eq:cotype-type}
C_r(X^*) \geq \frac{c\,T_p(X)}{1+\log\dim(X)}.
\eeq
Combining~\eqref{eq:type-qN} and~\eqref{eq:cotype-type} with~$\log\dim(\mathcal L(N;{\bf p})) = q\log N$, we thus obtain~\eqref{eq:cotype-3N}.

\paragraph{Open questions.}
For proving LDC lower bounds it is of interest to know what is the largest~$d$ such that~$\mathcal L(N;{\bf p})$ contains an $O(1)$-isomorphic copy of~$\ell_1^d$ when~$q \geq 3$.
For this purpose it in fact suffices to restrict to copies of~$\ell_1^d$ spanned by the type of forms appearing in the proof of Theorem~\ref{thm:main} below, which may be seen as lying in a generalization of the Birkhoff polytope (the set of doubly stochastic matrices).
Another question is if there is a converse to Theorem~\ref{thm:main}: Can a copy of~$\ell_1^k$ inside~$\mathcal L(N;{\bf p})$ be turned into an LDC-like object?

\paragraph{Outline.}
In Section~\ref{sec:prelims} we set a few notational conventions and gather some basic facts of normed spaces and Fourier analysis over the boolean hypercube.
In Section~\ref{sec:main} we prove the main result, Theorem~\ref{thm:main}.
In Section~\ref{sec:ldcs} we give alternative proofs for lower bounds on 2-query LDCs.
In the Appendix we combine similar ideas with a reduction inspired by~\cite{Kerenidis:2004} to give alternative proofs for lower bounds on LDCs with more queries.

\paragraph{Acknowledgements.}
I thank Oded Regev for inspiring conversations and useful comments on an earlier version of this manuscript, and I thank Mark Kim for helpful discussions early on.

\section{Preliminaries}
\label{sec:prelims}

\paragraph{Notation.}
For a positive integer~$n$ denote~$[n] = \{1,\dots,n\}$.
Denote by $B_{n,d}\subseteq\bset{n}$ the Hamming ball of radius~$d$ around the origin.
For a finite set~$S$ denote by~$\Exp_{x\in S}$ the expectation with respect to a uniformly distributed random element~$x$ in~$S$.
For a probability distribution~$\mu$ denote by $\Exp_{x\sim\mu}$ the expectation with respect to a random variable with distribution~$\mu$.
For sets~$\Gamma$ and~$S$, a positive integer~$r$, and a pair of ordered tuples~$z\in \Gamma^S$ and~${\bf S} = (s_1,\dots,s_r)\in S^r$, we denote by~$z_{\bf S}\in \Gamma^r$ the ordered tuple $(z_{s_1},\dots,z_{s_r})$.
With some abuse of notation we will apply set operations to ordered tuples: for~$S,{\bf S}$ as above write~$s\in {\bf S}$ if $s = s_j$ for some~$j\in [r]$ and write~$S\cap {\bf S}$ for the set~$\{s\in S\st  s_j =s \:\:\text{for some~$j\in[r]$}\}$.

\paragraph{Norms and spaces.}
For $1\leq p<\infty$, the~$\ell_p$-norm of a vector~$u\in \R^N$ is defined by
\beqn
\|u\|_{\ell_p} = \left(\sum_{i=1}^N |u_i|^p\right)^{1/p}.
\eeqn
Moreover, $\|u\|_{\ell_\infty} = \max_{i\in[N]}\{|u_i|\}$.
For~$p\in[1,\infty]$ denote by~$\ell_p^N$ the Banach space~$(\R^N, \|\:\:\|_{\ell_p})$.
For a finite set~$S$ we denote by~$\ell_q(S) = (\R^S, \|\:\:\|_{\ell_p})$ the space of vectors indexed by~$S$ endowed with the~$\ell_p$ norm.

\paragraph{Fourier analysis over the boolean hypercube.}
For a positive integer~$n$,  the $n$-dimensional boolean hypercube, denoted~$H_n$, is the group formed by the set~$\bset{n}$ endowed with entry-wise addition modulo~2.
The character group of~$H_n$ is formed by the functions~$\chi_u:H_n\to\R$ given by $\chi_u(x) = (-1)^{u\cdot x}$ for each $u\in\bset{n}$, where~$u\cdot x = u_1x_1 + \cdots + u_nx_n$.
A character~$\chi_u$ has \emph{degree}~$d$ if the string~$u$ has Hamming weight~$d$.
The character functions form a complete orthonormal basis for the Hilbert space of functions $f:H_n\to\R$ endowed with the inner product 
\beq\label{eq:bool-ip}
\langle f,g\rangle = \Exp_{x\in H_n}\big[f(x)g(x)\big].
\eeq
The Fourier transform~$\widehat{f}: H_n\to\R$ of a function~$f: H_n\to\R$ is given by~$\widehat{f}(u) = \langle f, \chi_u\rangle$.
A function~$f$ has \emph{degree~$d$} if its Fourier transform is supported by~$B_{n,d}$.
Orthogonality of the character functions with respect to the inner product~\eqref{eq:bool-ip} easily gives the {\em Fourier inversion formula}
\beqn
f(x) = \sum_{u\in H_n}\widehat f(u) \chi_u(x)
\eeqn
and \emph{Parseval's identity}
\beqn
\sum_{u\in H_n}\widehat f(u)^2 = \Exp_{x\in H_n}\big[f(x)^2\big].
\eeqn
It also follows easily from the above that a function~$f$ depends only on a subset $S\subseteq [n]$ of its variables if and only if~$\widehat f(u) = 0$ for every~$u\in H_n$ such that~$u_j =1$ for some~$j\not\in S$.
In particular, such a function has degree~$|S|$.

The above extends to Cartesian products of~$ H_n$, since for positive integers~$q$, we have~$H_n^q \cong H_{qn}$.
The characters of~$H_n^q$ are given by~$\chi_{\bf u} = \chi_{u[1]}\cdots\chi_{u[q]}$ for every~${\bf u} = (u[1],\dots,u[q])\in H_n^q$ and a function $f: H_n^q\to \R$ has \emph{degree~$d$} if its Fourier transform is supported by~$(B_{n,d})^q$.

%
\section{Copies of~$\ell_1^k$ from LDCs}
\label{sec:main}

In this section we prove Theorem~\ref{thm:main}.
In the restatement below, we use the fact that at a loss of at most a factor of~$2$ in~$|\Gamma|$, we may assume that~$\Gamma = H_n$ for some positive integer~$n$.
Also, for convenience later on, we will switch the message alphabet from~$\bset{}$ to~$\pmset{}$.

\begin{theorem}\label{thm:main2}
Let $\delta,\eps\in (0,1/2]$ and~$k,N,q,m,n$ be positive integers such that $q\geq 2$ and~${n\geq m}$. 
Assume there exists a $(q,\delta,\eps)$-LDC given by a map $C:\pmset{k}\to H_n^N$.
In addition assume that~$C$ has a decoder that uses at most~$m$ predetermined bits of each queried codeword symbol.
Then, 
for any vector
${\bf p} \in (1,\infty)^q$ such that $1/p_1 + \cdots + 1/p_q \leq 1$,
integer
$N' \geq {n\choose \leq m}N$, and real number
\beqn
K \geq \frac{q{n\choose \leq m}^{(q+2)/2}}{2\delta\eps},
\eeqn
the space $\mathcal L(N';{\bf p})$ contains  a $K$-isomorphic copy of~$\ell_1^k$.
\end{theorem}

For the rest of this section, let $k,N,q,m, n,\delta, \eps,{\bf p}$ be as in Theorem~\ref{thm:main2}.

\subsection{Smooth decoding}

The proof of Theorem~\ref{thm:main2} relies on a variant of a result of~\cite{Katz:2000}. Qualitatively the result says that an LDC allows us to retrieve any message bit with high probability from an uncorrupted codeword by sampling~$q$-tuples of codeword coordinates from a ``smooth'' distribution, in which the marginal distribution over single coordinates is roughly uniform.

\begin{lemma}\label{lem:smooth}
Let~$C:\pmset{k} \to H_n^N$ be a $(q,\delta,\eps)$-LDC.
Then, for each ${i\in[k]}$ there exists a probability distribution~$\mu_i$ over $[N]^q$ and for each ${{\bf S}\in[N]^q}$ there exists a function $f_{\bf S}^i:H_n^q \to [-1,1]$ such that:

\begin{itemize}
\item For every~$x\in\pmset{k}$, we have
$
x_i\, \Exp_{{\bf S}\sim\mu_i}\big[f_{{\bf S}}^i\big(C(x)_{{\bf S}}\big)\big] \geq 2\eps.
$

\item For every~$s\in [N]$, we have~$\Pr_{{\bf S}\sim\mu_i}[s\in {\bf S}] \leq 2q/(\delta N)$.
\end{itemize}

Moreover, if the LDC decoder's output depends on at most~$m$ predetermined bits of each queried codeword symbol, then~$f_{\bf S}^i$ has degree at most~$m$.
\end{lemma}

\begin{proof}
Fix an~$i\in[k]$.
Let~$\nu_i$ be a probability distribution over sets~$S\subseteq[N]$ of cardinality at most~$q$ and for every set~$S$ in the support of~$\nu_i$ let~$\phi_S$ be a map from~$H_n^S$ to the set of~$\pmset{}$-valued random variables.
Suppose that upon receiving the index~$i$ and a string~$y\in H_n^N$, the decoder samples a set~$S$ from~$\nu_i$ and outputs the random variable~$\phi_S(y_S)$.

Let~$S$ be a random set with distribution~$\nu_i$.
Let~$B\subseteq[N]$ be the set of \emph{bad} coordinates~$s\in [N]$ satisfying~$\Pr[s\in S] \geq q/(\delta N)$.
Since~$\nu_i$ is supported only on sets of size at most~$q$, it follows that~$|B|\leq \delta N$.
Let~$\widetilde{\bf S} = (\tilde s_s)_{s\in S}$ be the random \emph{sequence} such that for each bad coordinate~$s\in S$, the entry~$\tilde s_s$ is independent and uniformly distributed over~$[N]$ and for the other coordinates, we set~$\tilde s_s = s$.
We claim that, similar to the second item in the lemma, for every~$s\in[N]$, we have
\beq\label{eq:smooth}
\Pr[s\in\widetilde{\bf S}] \leq \frac{2q}{\delta N}.
\eeq
Indeed, for~$s\in [N]\smallsetminus B$, the probability in~\eqref{eq:smooth} is at most~$\Pr[s\in S]\leq q/(\delta N)$ plus the probability that~$s$ appears in a bad coordinate of~$\widetilde{\bf S}$.
By independence, the latter probability is at most~$q/N$, showing~\eqref{eq:smooth} for~$[N]\smallsetminus B$.
Bad elements ${s\in B}$ only appear at bad coordinates of~$\widetilde{\bf S}$. By independence, such elements therefore appear with probability at most~$q/N$, giving the claim.

Let~$\widetilde{S} = \{\tilde s_s\st s\in S\}$ be the random \emph{set} of distinct entries of~$\widetilde{\bf S}$ and let~$\tilde{\nu}_i$ be the distribution of~$\widetilde S$.
Let $T\subseteq[N]$ be a set of cardinality at most~$q$.
Recall from our notational convention (see Section~\ref{sec:prelims}) that for a vector $z\in H_n^T$, conditioned on the event~$\widetilde S = T$, the vector $z_{\widetilde{\bf S}} = (z_{\tilde s_s})_{s\in S}$ is well-defined as one lying in~$H_n^S$.
This allows us to define a function~$g_T^i:H_n^T\to[-1,1]$ by
\beq\label{eq:gTdef}
g_T^i(z) = \Exp\big[\phi_{S}^i (z_{\tilde{\bf S}}) \:|\: \widetilde{S} = T \big],
\eeq
where the expectation is taken over the set~$S$, the sequence~$\widetilde{\bf S}$, and the random value in~$\pmset{}$ assumed by the function~$\phi_S^i$.
We show that these functions~$g_T^i$ satisfy an inequality similar to the first item in the lemma, namely, we show that for every~$x\in\pmset{k}$ and random set~$T$ with distribution~$\tilde{\nu}_i$, we have
\beq\label{eq:gT2eps}
x_i\, \Exp_{T\sim \tilde{\nu_i}}\big[ g_T^i\big(C(x)_T\big) \big] \geq 2\eps.
\eeq
To show this, consider the random string~$y\in H_n^N$ 
where for every~$s\in[N]\smallsetminus B$, we have~$y_s = C(x)_s$ and for every~$s\in B$, we set~$y_s = C(x)_{t_s}$ where~$t_s$ is independent and uniformly distributed over~$[N]$.
As such,~$y$ is thus a random ``corrupted'' version of~$C(x)$ in which at most~$|B|\leq \delta N$ entries are replaced with other entries of the codeword.
We claim that the random sequences~$C(x)_{\widetilde{\bf S}}$ and~$y_{S}$ have the same distribution.
Indeed, observe that we get the first sequence if we sample~$S$ and then corrupt the sequence~$C(x)_S$  by 
replacing its entries at bad coordinates~$s$ by the random value $C(x)_{t_s}$ for~$t_s$ as above.
The second sequence $y_S$ corresponds to doing things in reverse order: first corrupt $C(x)$, giving~$y$, and then sample~$S$.
The claim follows because the values of the corrupted entries in~$S$ are independent of~$S$.
It follows that the random variables~$\phi_S^i\big(C(x)_{\widetilde{\bf S}}\big)$ and $\phi_S^i(y_{S})$ also have the same distribution and,
since~$y$ differs from~$C(x)$ in at most~$\delta N$ coordinates,
\beq\label{eq:equal-dist}
\Pr\big[\phi_S^i\big(C(x)_{\widetilde{\bf S}}\big) = x_i\big]
=
\Pr\big[\phi_S^i(y_{S}) = x_i\big]
\geq
\frac{1}{2} + \eps.
\eeq
Hence, since~$\widetilde S$ has the distribution~$\tilde{\nu}_i$, we have
\begin{align*}
x_i\, \Exp_{T\sim\tilde{\nu}_i}\big[g_{T}^i\big(C(x)_T\big)\big]\:
\stackrel{\,\eqref{eq:gTdef}\,\,}{=} &\:\:
x_i\, \Exp_{T\sim\tilde{\nu}_i}\Big[\Exp\big[\phi_S^i\big(C(x)_{\widetilde{\bf S}}\big) \: |\: \widetilde{S} = T \big] \Big]\\
\stackrel{\phantom{\eqref{eq:equal-dist}}}{=} &\:\:
x_i\,\Exp\big[\phi_S^i\big(C(x)_{\widetilde{\bf S}}\big)\big]\\
\stackrel{\eqref{eq:equal-dist}}{=} &\:\:
x_i\,\Exp\big[\phi_S^i(y_S)\big]\\
\stackrel{\phantom{\eqref{eq:equal-dist}}}{\geq} &\:\:
2\eps,
\end{align*}
where the first inner expectation and the second and third expectations are taken over the set~$S$, the sequence~$\widetilde{\bf S}$, the set~$\widetilde S$, and the random value of the function~$\phi_S^i$.
This shows~\eqref{eq:gT2eps}.


Define the probability distribution~$\mu_i:[N]^q\to [0,1]$ as follows.
For a set~$T\subseteq[N]$ with cardinality at most~$q$, let~$\mathcal F(T)\subseteq T^q$ be the family of ordered sequences that contain each element of~$T$ at least once.
For each ${\bf T}\in \mathcal F(T)$ set~$\mu_i({\bf T}) = \tilde{\nu}_i(T)/|\mathcal F(T)|$.
Then, by~\eqref{eq:smooth} we have
\beqn
\Pr_{{\bf T}\sim\mu_i}[s\in{\bf T}] 
= 
\Pr_{T\sim\tilde{\nu}_i}[s\in T] \leq \frac{2q}{\delta N}
\eeqn
for each~$s\in [N]$.
For each set~$T$ in the support of~$\tilde{\nu}_i$ and every~${\bf T}\in\mathcal F(T)$, there exists a function $f_{\bf T}^i : H_n^q \to [-1,1]$ such that
$
f_{\bf T}^i\big(z_{\bf T}\big)
=
g_T^i(z)
$
holds for each~$z \in H_n^T$ (as~$z_{\bf T}$ and~$z$ have entries from the same set).
Pick one such function arbitrarily.
For all remaining~${\bf T}\in [N]^q$ let~$f_{\bf T}^i$ be identically zero.
By~\eqref{eq:gT2eps}, these functions satisfy the first item of the lemma.

Finally, observe that if the decoder's output depends on at most~$m$ predetermined bits of each queried codeword symbol, then for each set~$T$ in the support of~$\nu_i$, the function~$h_T^i:H_n^T\to[-1,1]$ defined by~$h_T^i(z) = \Exp[\phi_T^i(z)]$, where the expectation is taken over the randomness in~$\phi_T^i$, has degree at most~$m$.
Since the functions~$g_T^i$ in~\eqref{eq:gTdef} are linear combinations of these~$h_T^i$, they also have degree at most~$m$.
It follows that the functions~$f_{\bf T}^i$ can be chosen to satisfy the same.
\end{proof}

\subsection{Norms of some forms}

The proof of Theorem~\ref{thm:main} uses the functions~$f_{\bf S}^i$ and distributions~$\mu_i$ of Lemma~\ref{lem:smooth} to construct a basis~$A_1,\dots,A_k\in\mathcal L(N';{\bf p})$ for a copy of~$\ell_1^k$ as in~\eqref{eq:l1copy}.
Viewed as a $q$-tensor, the form~$A_i$ will consist of blocks, one block for each $q$-tuple~${\bf S}\in [N]^q$, and the entries of each block will contain the Fourier coefficients of the function~$f_{\bf S}^i$ scaled by the probability~$\mu_i({\bf S})$.
We use the following facts to show that these forms have the desired properties.

\begin{proposition}\label{prop:ftensor}
Let~$f:H_n^q\to[-1,1]$ be a function of degree at most~$m$.
Define the~$q$-linear form~$F$ on $\R^{B_{n,m}}$ by
\beqn
F({\bf z})
=
\sum_{{\bf u}\in B_{n,m}^q} \widehat{f}({\bf u})\, z[1]_{u[1]}\cdots z[q]_{u[q]}
\eeqn
for~${\bf z} = (z[1],\dots,z[q]) \in \R^{B_{n,m}}\times\cdots\times\R^{B_{n,m}}$.
Then, $\|F\|_{\bf p} \leq |B_{n,m}|^{q/2}$.
\end{proposition}

\begin{proof}
H\"older's inequality implies that a vector in the unit ball of~$\ell_p^t$ has~$\ell_2$-norm at most~$t^{1/2 - 1/p}$.
Hence, by the Cauchy-Schwarz inequality and Parseval's identity,
\begin{align*}
|F({\bf z})|
&\leq
\left(\sum_{{\bf u}\in B_{n,m}^q} \widehat{f}({\bf u})^2\right)^{1/2}\,
\prod_{j=1}^q \|z[j]\|_{\ell_2}\\
&=
\prod_{j=1}^q \|z[j]\|_{\ell_2} \\
&\leq
\prod_{j=1}^q |B_{n,m}|^{1/2 - 1/p_j}\|z[j]\|_{\ell_{p_j}}\\
&\leq 
|B_{n,m}|^{q/2} \prod_{j=1}^q \|z[j]\|_{\ell_{p_j}}.
\end{align*}
\end{proof}


We use a generalization of a doubly-stochastic matrix.
Let~${\bf 1}\in \R^N$ denote the all-ones vector.

\begin{definition}[Plane sub-stochastic form]\label{def:plane-stoch}
A $q$-linear form~$A$ on~$\R^N$ is \emph{plane sub-stochastic} if 
the tensor
$
T = \big(A(e_{s_1},\dots,e_{s_q})\big)_{{\bf S}\in [N]^q}
$
is nonnegative and
for every $s\in[N]$, we have
\begin{align}
A(e_s, {\bf 1},{\bf 1},\dots,{\bf 1}) &\leq 1 \nonumber\\
A({\bf 1}, e_s,{\bf 1},\dots,{\bf 1}) &\leq 1 \nonumber\\
&\vdots \nonumber\\
A({\bf 1}, {\bf 1},\dots,{\bf 1},e_s) &\leq 1. \label{eq:stoch}
\end{align}
\end{definition}

\begin{remark}
The above definition gives the Birkhoff polytope when we set~$q = 2$ and we change the inequalities in~\eqref{eq:stoch} to equalities.
Recall that the Birkhoff--von~Neumann Theorem states that the Birkhoff polytope is the convex hull of the set of $N\times N$ permutation matrices.
Interestingly, Linial and Luria~\cite{Linial:2014} showed that for $q\geq 3$, the polytope of  $q$-linear ``plane stochastic forms'' corresponding to equalities in~\eqref{eq:stoch} is \emph{not} contained in the convex hull of the set of ``permutation tensors'' defined as 0/1 tensors satisfying equality in~\eqref{eq:stoch}.
\end{remark}

\begin{proposition}\label{prop:stoch}
If~$A\in \mathcal L(N;{\bf p})$ is plane sub-stochastic, then $\|A\|_{\bf p} \leq 1$.
\end{proposition}

The proof uses the following result of Carlen, Loss, and Lieb~\cite{Carlen:2006}.

\begin{theorem}[Multi-linear Riesz--Thorin Interpolation Theorem]\label{thm:CLL}
Let~$A$ be a~$q$-linear form on~$\R^N$.
Let~$\psi: [0,1]^q\to \R_+$ be the function defined by
$
\psi(1/r_1,\dots,1/r_q) = \|A\|_{\bf r},
$
for any~${\bf r}\in [1,\infty]^q$.
Then, $\ln(\psi)$ is a convex function on~$[0,1]^q$.
\end{theorem}

\begin{proof}[Proof of Proposition~\ref{prop:stoch}]
We first show that
$\|A\|_{\bf r} \leq 1$
for any~${\bf r}$ consisting of one~$1$ entry and all the others~$\infty$.
Indeed, by H\"older's inequality and the assumption that~$|A|$ is plane sub-stochastic,
\begin{align*}
|A(z[1],\dots,z[q])|
&\leq
\sum_{{\bf S}\in [N]^q} |A(e_{s_1}, \dots e_{s_q})z[1]_{s_1}\cdots z[q]_{s_q}|\\
&\leq
\Big(\prod_{j=1}^{q-1}\|z[j]\|_{\ell_\infty}\Big) \sum_{s=1}^N A({\bf 1}, \dots, {\bf 1}, s)|z[q]_s|\\
&\leq
\Big(\prod_{j=1}^{q-1}\|z[j]\|_{\ell_\infty}\Big)\|z[q]\|_{\ell_1}.
\end{align*}
Hence, if~${\bf r} = (\infty,\dots,\infty,1)$, we have~$\|A\|_{\bf r} \leq 1$.
The cases for the other positions of the 1-entry are proved in the same way.
Since for these choices of~${\bf r}$, the vectors $(1/r_1,\dots,1/r_q)$ are the $q$ standard basis vectors, the vector~${\bf p}$ lies in their convex hull, and the result follows from Theorem~\ref{thm:CLL}.
\end{proof}

\subsection{Proof of the main result}

With this, the proof of Theorem~\ref{thm:main2} is straightforward.

\begin{proof}[Proof of Theorem~\ref{thm:main2}]
For each index $i\in[k]$ and ${\bf S}  = (s_1,\dots,s_q) \in[N]^q$ let~$\mu_i$ and~$f_{\bf S}^i$ be a distribution and function as in Lemma~\ref{lem:smooth}.
Note that the Fourier transform of each~$f_{\bf S}^i$ is supported by the Cartesian product of Hamming balls~$(B_{n,m})^q$.
Define a~$q$-linear form~$F_{\bf S}^i$ on~$\R^{B_{n,m}}$ based on the Fourier coefficients of the function~$f_{\bf S}^i$ as in Proposition~\ref{prop:ftensor}.
For a vector $z\in \R^{[N]\times B_{n,m}}$ and $s\in[N]$ write~$z_s$ for the projection of~$z$ onto the coordinates $(s,u)$ with~$u\in B_{n,m}$, that is,
$
z_s = (z_{(s, u)})_{u\in B_{n,m}} \in \R^{B_{n,m}}.
$
For a tuple ${\bf z} = (z[1],\dots,z[q])$ with each~$z[j]\in \R^{[N]\times B_{n,m}}$, write ${\bf z}_{\bf S} = (z[1]_{s_1},\dots,z[q]_{s_q})$.
Let~$A_i$ be the~$q$-linear form on~$\R^{[N]\times {B_{n,m}}}$ defined by
\beq\label{eq:Mdef}
A_i({\bf z})
=
\Exp_{{\bf S}\sim \mu_i}\big[ F_{\bf S}^i({\bf z}_{\bf S}) \big]
=
\Exp_{{\bf S}\sim \mu_i}\Big[ \sum_{{\bf u}\in (B_{n,m})^q} \widehat{f_{\bf S}^i}({\bf u})\,  z[1]_{(s_1, u[1])}\cdots z[q]_{(s_q, u[q])} \Big].
\eeq

Fix a vector~$\alpha\in\R^k$, let~$x = (\sign(\alpha_i))_{i=1}^k$ and let~$y= C(x)$ be the codeword in~$H_n^N$ corresponding to the message~$x$.
Define the sign vector
$
\widehat y = (\chi_u(y_s))_{s\in[N], u\in {B_{n,m}}}.
$
Let~$\widehat{{\bf y}} = (\widehat y, \dots, \widehat y)$ ($q$ times).
By the Fourier Inversion Formula, 
\begin{align}
F_{\bf S}^i\big(\widehat {\bf y}_{\bf S}\big)
\:=\:
\sum_{{\bf u}\in (B_{n,m})^q} \widehat{f_{\bf S}^i}({\bf u})\, \chi_{\bf u}(y_{\bf S})
\:=\:
f_{\bf S}^i\big(y_{\bf S}\big)
\:=\:
f_{\bf S}^i\big(C(x)_{\bf S}\big)
.\label{eq:FCval}
\end{align}
Combining~\eqref{eq:FCval} with the first property of the~$f_{\bf S}^i$ in Lemma~\ref{lem:smooth} then gives
\beq\label{eq:aD}
\alpha_i\, \Exp_{{\bf S}\sim \mu_i}\big[F_{\bf S}^i\big(\widehat {\bf y}_{\bf S}\big)\big]
\geq 
2|\alpha_i|\eps.
\eeq
Since~$\widehat y$ is a sign vector of dimension~$N|B_{n,m}|$, it has~$\ell_p$-norm $(N|B_{n,m}|)^{1/p}$. 
Normalizing accordingly and using $q$-linearity of the~$A_i$, we get
%
\begin{align}
\Big\|\sum_{i=1}^k\alpha_iA_i\Big\|_{\bf p} 
&\geq
\frac{1}{(N|B_{n,m}|)^{1/p_1 + \cdots + 1/p_q}}\left(\sum_{i=1}^k\alpha_i A_i\right)\big(\widehat {\bf y}\big) \nonumber\\
&\stackrel{\eqref{eq:Mdef}}{\geq}
\frac{1}{N|B_{n,m}|} \sum_{i=1}^k\alpha_i \Exp_{{\bf S}\sim \mu_i}\big[F_{\bf S}^i\big(\widehat {\bf y}_{\bf S}\big)\big] \nonumber\\
&\stackrel{\eqref{eq:FCval}, \eqref{eq:aD}}{\geq}
\frac{2\eps}{N|B_{n,m}|} \|\alpha\|_{\ell_1}.\label{eq:aMnorm}
\end{align}

Next, we bound the norms of the forms~$A_i$ themselves.
Let ${\bf z} = (z[1],\dots,z[q])$ be a $q$-tuple of nonzero vectors in~$\R^{[N]\times B_{n,m}}$.
Recall from Proposition~\ref{prop:ftensor} that each~$F_{\bf S}^i$ has norm at most~$|B_{n,m}|^{q/2}$.
This implies
\begin{align}
\big|A_i({\bf z})\big| 
&\stackrel{\eqref{eq:Mdef}}{\leq}
\Exp_{{\bf S}\sim\mu_i}\Big[ \big|F_{\bf S}^i({\bf z}_{\bf S}) \big|\Big]\nonumber\\
&\stackrel{\phantom{\eqref{eq:Mdef}}}{\leq}
|B_{n,m}|^{q/2}\, \Exp_{{\bf S}\sim\mu_i}\Big[ \|z[1]_{s_1}\|_{\ell_{p_1}}\cdots \|z[q]_{s_q}\|_{\ell_{p_q}} \Big].\label{eq:Minorm}
\end{align}
To bound the above expectation define the $q$-tuple ${\bf a} = (a[1],\dots,a[q])$ of (nonnegative) vectors
\beq\label{eq:avecdef}
a[j] = \Big(\|z[j]_s\|_{\ell_{p_j}}\Big)_{s=1}^N,
\quad
j\in[q].
\eeq
Define the~$q$-linear form~$M$ on $\R^N$ given by
$
M({\bf b}) 
= 
\Exp_{{\bf S}\sim\mu_i}[b[1]_{s_1}\cdots b[q]_{s_q}].
$
Then, the expectation in~\eqref{eq:Minorm} equals~$M({\bf a})$.
The form~$M$ is clearly nonnegative and by the second item in Lemma~\ref{lem:smooth}, the scaled version $(\delta N/q) M$ is plane sub-stochastic since for each~$t\in [N]$ and~$j\in[q]$, we have
\beqn
M[\underbrace{{\bf 1},\dots, {\bf 1}}_{1,\dots,j-1}, \underbrace{e_t}_{j}, \underbrace{{\bf 1}, \dots, {\bf 1}}_{j+1,\dots,q}] = 
\Exp_{{\bf S}\sim\mu_i}\big[(e_t)_{s_j}\big] =
\Pr_{{\bf S}\sim\mu_i}[s_j = t] \leq \frac{q}{\delta N}.
\eeqn
By Proposition~\ref{prop:stoch}, the form $M$ therefore has norm at most~$\|M\|_{\bf p} \leq q/(\delta N)$.
Since each~$a[j]$ as in~\eqref{eq:avecdef} has norm $\|a[j]\|_{\ell_{p_j}} = \|z[j]\|_{\ell_{p_j}}$, we conclude that each~$A_i$ has norm $\|A_i\|_{\bf p} \leq q|B_{n,m}|^{q/2}/(\delta N)$.

Hence, for any vector~$\alpha\in\R^k$, by~\eqref{eq:aMnorm} and the triangle inequality,
\beqn
\frac{2\eps}{N|B_{n,m}|}\|\alpha\|_{\ell_1}
\leq 
\Big\|\sum_{i=1}^k\alpha_iA_i\Big\|_{\bf p} 
\leq 
\frac{q|B_{n,m}|^{q/2}}{\delta N}\|\alpha\|_{\ell_1}.
\eeqn
Scaling the~$A_i$ by $N|B_{n,m}|/(2\eps)$ then gives the copy of~$\ell_1^k$ as desired.
\end{proof}

\section{Lower bounds on 2-query LDCs}
\label{sec:ldcs}

In this section we use Theorem~\ref{thm:main} to prove the 2-query LDC lower bounds mentioned in the Introduction (up-to slightly poorer dependence on~$\delta$ and~$|\Gamma|$).
%
The key is the following bound on the dimension~$k$ for which~$\mathcal L(N;(2,2))$ can accommodate a copy of~$\ell_1^k$ with distortion~$K$.
The bound is surely well-known, but it does not appear to be published in the form below.

\begin{lemma}\label{lem:l1Sinf}
There exists an absolute constant~$C\in (0,\infty)$ such that the following holds.
Suppose that for~$K<\infty$ the space~$\mathcal L(N;(2,2))$ contains a~$K$-isomorphic copy of~$\ell_1^k$.
Then
$k \leq CK^2\log (2N)$.
\end{lemma}

Lemma~\ref{lem:l1Sinf} follows easily from a random-matrix inequality belonging to a family of ``non-commutative Khintchine inequalities'' due to Tomczak-Jaegermann~\cite{Tomczak-Jaegermann:1974} (not to be confused with the stronger non-commutative Khintchine inequalities of Lust-Piquard and Pisier~\cite{Lust-Piquard:1991}).
Recall that the Schatten-$\infty$ norm $\|A\|_{S_\infty}$ of a matrix ${A\in\R^{N\times N}}$ is  the supremum of~$|u^{\mathsf T}Av|/\|u\|_2\|v\|_2$ over nonzero vectors~$u,v\in\R^N$.

\begin{theorem}[Tomczak-Jaegermann]\label{thm:Sinf-type}
There exists an absolute constant~$C\in (0,\infty)$ such that the following holds.
Let~$N$ and~$k$ be positive integers, let~$A_1,\dots,A_k \in \R^{N\times N}$, and let $\epsilon_1,\dots, \epsilon_k$ be independent uniformly distributed $\pmset{}$-valued random variables.
Then,
\beq\label{eq:Sinf-type}
\Exp\Big[\Big\|\sum_{i=1}^k \epsilon_iA_i\Big\|_{S_\infty}\Big]
\leq
C\sqrt{\log (2N)}\, \Big(\sum_{i=1}^k \|A_i\|_{S_\infty}^2\Big)^{1/2}.
\eeq
\end{theorem}

\begin{remark}
To extract the above from~\cite[Theorem~3.1]{Tomczak-Jaegermann:1974} we used the standard and easy fact that for~$p = \log N$, the Schatten-$p$ norm of a matrix~$A\in \R^{N\times N}$, defined as the~$\ell_p$-norm of the vector of singular values of~$A$, satisfies~$\|A\|_{S_\infty}\leq \|A\|_p \leq C\|A\|_{S_\infty}$ for some absolute constant~$C\in [1,\infty)$.
\end{remark}

\begin{remark}
Similar (stronger) estimates were proved in~\cite{Lust-Piquard:1991, Buchholz:2005, Oliveira:2010b, Tropp:2012}.
\end{remark}

\begin{proof}[Proof of Lemma~\ref{lem:l1Sinf}]
Identify the space~$\mathcal L(N;(2,2))$ with $(\R^{N\times N}, \|\:\:\|_{S_\infty})$.
Let~$A_1,\dots,A_k\in \R^{N\times N}$ be matrices such that~\eqref{eq:l1copy} holds (with $X = S_\infty$).
Setting the vector~$\alpha$ in~\eqref{eq:l1copy} to be a standard basis vector we see that~$\|A_i\|_{S_\infty} \leq K$ for each~$i\in[k]$.
Hence, by~\eqref{eq:l1copy} and Theorem~\ref{thm:Sinf-type},
\beqn
k
\leq
\Exp\Big[ \Big\|\sum_{i=1}^k \epsilon_iA_i\Big\|_{S_\infty} \Big]
\leq
CK\sqrt{k\log(2N)}.\qedhere
\eeqn
\end{proof}

Theorem~\ref{thm:main} asserts that a $(2,\delta,\eps)$-LDC from $\bset{k}$ to~$\Gamma^N$ gives a $4|\Gamma|^2/(\delta\eps)$-isomorphic copy of~$\ell_1^k$ in the space~$\mathcal L(2|\Gamma|N; (2,2))$.
Combining this with Lemma~\ref{lem:l1Sinf} immediately gives the following exponential lower bounds on binary 2-query LDCs.

\begin{corollary}
Any binary $(2,\delta,\eps)$-LDC satisfies $N \geq 2^{\Omega(\delta^2\eps^2 k)}$.
\end{corollary}

For LDCs over larger alphabets we obtain the following bound.

\begin{corollary}
Any $(2,\delta,\eps)$-LDC with~$\Gamma = H_n$ and a decoder that uses at most~$m$ out of~$n$ predetermined bits of each queried codeword symbol satisfies
\beq\label{eq:largeG}
{n\choose \leq m}^3\left(\log N + \log{n\choose \leq m}\right) \geq \Omega(\delta^2\eps^2 k).
\eeq
\end{corollary}

\begin{remark}
A more careful analysis in Section~\ref{sec:main} for the case $1/p_1 + \cdots + 1/p_q=1$ allows one to replace the third power in the left-hand side of~\eqref{eq:largeG} with a square.
\end{remark}


%

\section{Lower bounds on LDCs with more queries}
\label{sec:qldcs}

Here we prove lower bounds on binary $q$-query LDCs for~$q \geq 3$ using a method
similar to the reductions to the two-query case used in~\cite{Kerenidis:2004}, but in the spirit of Theorem~\ref{thm:main}.

\begin{theorem}\label{thm:qldcMs}
Let~$\delta,\eps\in(0,1/2]$ and let~$k,l, N,r$ be positive integers such that~$r \geq 2$ and
\beq\label{eq:lbound}
\frac{\delta\eps}{(2r)^3}l^r \geq (2rN)^{r-1}.
\eeq
Suppose that there exists a~$(2r,\delta,\eps)$-LDC
from~$\pmset{k}$ to~$\pmset{N}$.
Then, there exist matrices $A_1,\dots,A_k\in\R^{N^l\times N^l}$ such that ${\|A_i\|_{S_\infty} \leq 1}$ for each ${i\in[k]}$, and for independent uniformly distributed $\pmset{}$-valued random variables $\epsilon_1,\dots, \epsilon_k$, we have
\beq\label{eq:avel1}
\Exp\Big[\Big\|\sum_{i=1}^k \epsilon_i A_i\Big\|_{S_\infty}\Big] \geq \frac{\eps k}{8^{r}}.
\eeq
\end{theorem}

Elton's Theorem asserts that the above is sufficient to find a finite-dimensional copy of~$\ell_1$ (see Vershynin and Mendelson~\cite[Theorem~3]{Mendelson:2003} for the quantitatively optimal form stated below).

\begin{theorem}[Elton's Theorem]
There exists a absolute constant~${c>0}$ such that the following holds.
Let~$X$ be a Banach space, let $A_1,\dots,A_k$ be vectors in the unit ball of~$X$, and let~$\eta>0$ be such that for independent uniformly distributed $\pmset{}$-valued random variables $\epsilon_1,\dots, \epsilon_k$, we have
\beqn
\Exp \Big[ \Big\|\sum_{i=1}^k\epsilon_i A_i\Big\|_X \Big] \geq \eta k.
\eeqn
Then, there exists a set~$I\subseteq[k]$ of cardinality~$|I|\geq c\eta^2 k$ such that for any $\alpha\in\R^I$, we have

\beqn
c\eta\|\alpha\|_{\ell_1} 
\leq
\Big\|\sum_{i\in I}\alpha_i A_i\Big\|_X
\leq
\|\alpha\|_{\ell_1}.
\eeqn
\end{theorem}

Combining Theorem~\ref{thm:qldcMs} with Elton's Theorem shows that, for a positive integer~$r\geq 2$, a $(2r,\delta,\eps)$-LDC gives a~$K$-isomorphic copy of~$\ell_1^d$ inside~$\mathcal L(N^l;(2,2))$ for $K = \delta(\eps/q)^2$ and~$d \geq ck/K^2$.
Through Lemma~\ref{lem:l1Sinf} this leads to a lower bound on $(2r)$-LDCs similar to the one stated in the Introduction.
However, combining Theorem~\ref{thm:qldcMs} with  Theorem~\ref{thm:Sinf-type} gives the following lower bound that has slightly better dependence on~$\delta$, $\eps$, and~$r$.

\begin{corollary}\label{cor:evenq}
For every integer~$r\geq 2$ there exists a~$c>0$ such that the following holds.
Suppose that for positive integers~$k$ and~$N$ and~$\delta,\eps\in(0,1/2]$, there exists a~$(2r,\delta,\eps)$-LDC
from~$\pmset{k}$ to~$\pmset{N}$.
Then,
$
N \geq c(\delta\eps^3k/\log k)^{r/(r-1)}.
$
\end{corollary}

\begin{proof}
Let~$l$ be the smallest integer satisfying~\eqref{eq:lbound} and let~$A_1,\dots,A_k$ be matrices as in Theorem~\ref{thm:qldcMs}.
Theorem~\ref{thm:Sinf-type} then gives
\beqn
\frac{\eps k}{8^r}
\leq
\Exp\Big[\Big\|\sum_{i=1}^k \epsilon_i A_i\Big\|_{S_\infty}\Big]
\leq
\sqrt{2k\log (2eN^l)}
\leq
c\sqrt{\frac{kN^{(r-1)/r}\log N}{(\delta\eps)^{1/r}}},
\eeqn
where~$c <\infty$ is a constant depending on~$r$ only.
Rearranging gives
$
N^{(r-1)/r}\log N \geq c' (\delta\eps)^{2/r}\eps k,
$
where~$c' > 0$ is a constant depending on~$r$ only, which implies the claim.
\end{proof}

The above bound is slightly poorer than the one stated in~\cite{Woodruff:2007a}, albeit only by a $\poly(\log k)$ factor.
It would be interesting to see if Elton's Theorem can be avoided in creating a copy of~$\ell_1^d$ inside~$\mathcal L(N;(2,2))$.

We proceed with the proof of Theorem~\ref{thm:qldcMs}, for which we use the following slight variant of a standard ``matching lemma'' of~\cite[Appendix~B]{Ben-Aroya:2008}, shown in~\cite{Briet:2012d} for~$q = 3$.
We omit the proof, which is a straightforward modification of~\cite[Lemma~3.1]{Briet:2012d}.

\begin{lemma}[Ben-Aroya--Regev--de~Wolf]\label{lem:matchings}
Let~$C:\pmset{k}\to \pmset{N}$ be a $(q,\delta,\eps)$-LDC.
Then, there exists a function~$C':\pmset{k}\to \pmset{qN}$ such that the following holds.
For every~$i\in[k]$ there exists a family~$\mathcal M_i$ of at least $\delta\eps N/q^2$ pairwise disjoint sets~$S\subseteq [qN]$ of~$q$ elements each, such that for a uniformly distributed random string~$x\in\pmset{k}$, we have
\beq\label{eq:match-exp}
\Big|\Exp\Big[ x_i\, \prod_{s\in S} C'(x)_{s}\Big]\Big| \geq \frac{\eps}{2^q}.
\eeq
\end{lemma}

We also use the following proposition, which may be interpreted as a generalization of the Birthday Paradox.

\begin{proposition}\label{prop:good-tuples}
For~$\eta > 0$ and positive integers~$N$ and~$r\geq 2$,
let $\mathcal F$ be a family of~$\eta N$ pairwise disjoint subsets $S\subseteq [N]$, each of cardinality~$|S| = 2r$.
Let~$l$ be a positive integer such that 
\beq\label{eq:l-assump}
\eta l^r \geq N^{r-1}.
\eeq
Then, there exists a set~$\mathcal I\subseteq [N]^l$ of cardinality at least~$N^l/4^{r}$ such that for each sequence ${\bf S}\in \mathcal I$, there exists an~$S\in\mathcal F$ for which~$|S\cap {\bf S}| \geq r$.
\end{proposition}


The proof of Proposition~\ref{prop:good-tuples} uses a standard Poisson approximation result for ``balls and bins'' problems \cite[Theorem~5.10]{Mitzenmacher:2005}.
A discrete Poisson random variable~$Y$ with expectation~$\mu$ is nonnegative, integer valued, and has probability density function
\beq\label{eq:poisson}
\Pr[Y = m] = \frac{e^{-\mu}\mu^m}{m!},
\quad\quad
\forall m=0,1,2,\dots
\eeq

\begin{theorem}[Poisson approximation]\label{thm:poisson}
For positive integers~$l$ and~$N$, suppose we toss~$l$ balls into~$N$ bins independently and uniformly at random.
For each~$s\in [N]$ let~$X_s$ be the random variable counting the number of balls in bin number~$s$.
Let~$Y_1,\dots,Y_N$ be independent Poisson random variables with expectation~$l/N$.
Then, for any function $f:\{0,\dots,l\}^N\to \R$ such that~$\Exp[f(X_1,\dots,X_N)]$ increases or decreases monotonically with~$l$, we have
$
\Exp[f(X_1,\dots,X_N)] \leq 2 \Exp[f(Y_1,\dots,Y_N)].
$
\end{theorem}


\begin{proof}[Proof of Proposition~\ref{prop:good-tuples}]
Let us assume for simplicity that~$N$ is a multiple of~$2r$.
Partition the elements in~$[N]$ not covered by any set in~$\mathcal F$ into disjoint sets of size~$2r$.
With~$\mathcal F$, this gives a partition~$\mathcal P$ of~$[N]$ into $M = N/(2r)$ sets of size~$2r$.
Label the sets in~$\mathcal F$ with distinct numbers in~$[\eta N]$ and label the remaining partitions with distinct numbers in~$\{\eta N+1,\dots, M\}$.

Let~$b_1,\dots,b_l$ be independent uniformly distributed random variables over $[N]$, the balls.
Say that ball~$b_j$ lands in bin~$S\in\mathcal P$ if $b_j\in S$ and notice that the balls land in a uniformly random bin.
For each~$s\in[M]$ let~$X_s$ be the random variable counting the number of balls in bin number~$s$
and let~$Y_s$ be a discrete Poisson random variable with expectation~$\mu = l/M$.

Let $f:\{0,1,2,\dots\}^{M}\to\bset{}$ be the function that assumes the value~$1$ if and only if its first~$\eta N$ variables have value strictly less than~$r$.
Clearly $\Exp[f(X_1,\dots,X_{M})]$ decreases monotonically with~$l$, the number of balls we toss, since this expectation equals the probability that all bins in~$\mathcal F$ have strictly less than~$r$ balls.
Therefore, by Theorem~\ref{thm:poisson}, we have
\beq\label{eq:poisson-approx}
\Pr[f(X_1,\dots,X_{M}) = 1]
\leq
2\Pr[f(Y_1,\dots,Y_{M}) = 1].
\eeq
Independence of the~$Y_j$ gives
\begin{align}
\Pr[f(Y_1,\dots,Y_{M}) = 1]
&= \prod_{j=1}^{\eta N} \Pr[Y_j < r] \nonumber\\
&= \left( \sum_{m=0}^{r-1}\frac{e^{-\mu}\mu^m}{m!}\right)^{\eta N} \nonumber\\
&= \left( 1 - e^{-\mu}\sum_{m=r}^\infty \frac{\mu^m}{m!} \right)^{\eta N} \nonumber\\
&\leq \left(1 - \frac{\mu^r}{r!}\right)^{\eta N},\label{eq:Pbound}
\end{align}
By our assumption~\eqref{eq:l-assump} on~$l$ and the easy bound~$r^r/r! \geq 1$, we have 
\beqn
\frac{\mu^r}{r!}
=
\frac{1}{r!}\Big( \frac{l}{M} \Big)^r
=
\frac{1}{r!}\Big( \frac{2rl}{N} \Big)^r
\geq \frac{1}{\eta N}.
\eeqn 
Hence, \eqref{eq:Pbound} is at most~$1/e$ and it follows from~\eqref{eq:poisson-approx} that with probability $1 - 2/e\geq 1/4$, one of the first~$\eta N$ bins has at least~$r$ balls.
In other words, for at least~$N^l/4$ sequences~${\bf S}\in [N]^l$ there exists an~$S\in\mathcal F$ such that~$r$ entries of~${\bf S}$ belong to~$S$.
Of those sequences, a $2r(2r-1)\cdots r/(2r)^r\geq (1/2)^r$ fraction has those entries distinct.
Since we assumed that~$r\geq 2$, at least~$N^l/(2^{2+r}) \geq N^l/4^r$ of the sequences have the desired property.
\end{proof}

\begin{proof}[Proof of Theorem~\ref{thm:qldcMs}]
Let~$C:\pmset{k}\to\pmset{N}$ be a~$(2r,\delta,\eps)$-LDC. 
Let~$C':\pmset{k}\to\pmset{2rN}$ be a map and~$\mathcal M_i$ be families as in Lemma~\ref{lem:matchings}. 
Let~$N' = 2rN$ and recall that each~$\mathcal M_i$ consists of at least~$\delta\eps N'/(2r)^3$ pairwise disjoint sets $S\subseteq[N']$ of cardinality~$2r$ each.
Let~$l$ be an integer such that~\eqref{eq:lbound} holds.
Fix an~$i\in[k]$.
By proposition~\ref{prop:good-tuples}, there is a set $\mathcal I_i\subseteq [N']^l$ of at least~$(N')^l/4^{r}$ sequences ${\bf S}\in [N']^l$ such that~$|S\cap {\bf S}| \geq r$ for some~$S\in\mathcal M_i$.

We define a partial matching~$\mathcal P_i$ in~$[N']^l$.
For each~${\bf S}\in\mathcal I_i$ and associated $S\in\mathcal M_i$, pick a set~$T\subseteq[l]$ of~$r$ coordinates such that~$|{\bf S}_T\cap S| = r$.
Let ${\bf S}'\in [N']^l$ be a sequence such that ${\bf S}_t' = {\bf S}_t$ for every~$t\not\in T$ and such that (with slight abuse of notation) ${\bf S}_T\cup {\bf S}'_T= S$.
There are~$r!$ choices of~${\bf S}_T'$. 
Choose one arbitrarily but uniquely.
Let~$\mathcal P_i$ be the family of said~$\{\bf S, S'\}$ pairs and observe that $|\mathcal P_i| = |\mathcal I_i|/2 \geq (N')^l)/4^{r}$.

Define a matrix~$A_i:[N']^l\times [N']^l \to \{-1,0,1\}$ as follows.
Let $x$ be a uniformly distributed random string over~$\pmset{k}$ and notice that~$x_1,\dots,x_k$  are independent and uniformly distributed over~$\pmset{}$.
For each pair $\{{\bf S}, {\bf S'}\}\in\mathcal P_i$ with associated set~$S\in \mathcal M_i$ as above, set
\beq\label{eq:ASSdef}
A_i({\bf S}, {\bf S'}) = \sign\Big(\Exp\Big[x_i\, \prod_{s\in S}C'(x)_{s} \Big]\Big).
\eeq
Set the other entries of~$A_i$ to zero.
Moreover, since~$\mathcal P_i$ is a (partial) matching, each row and column of~$A_i$ has at most one nonzero element and it follows that~$\|A_i\|_{S_\infty}\leq 1$.
For each~$x\in\pmset{n}$ let $D(x) = C'(x)^{\otimes l}$. 
Then, for each~$\{{\bf S} = (s_1,\dots,s_l), {\bf S'} = (s_1',\dots,s_l')\}\in\mathcal P_i$, with associated sets~$S\in\mathcal M_i$ and~$T\subseteq[l]$ as above, 
\beq\label{eq:DSS}
D(x)_{\bf S}D(x)_{\bf S'} = \Big(\prod_{t\not\in T}C'(x)_{s_t}^2 \Big)\prod_{t\in T}C'(x)_{s_t}C'(x)_{s_t'} = \prod_{s\in S}C'(x)_{s}.
\eeq
Hence, by Lemma~\ref{lem:matchings}, for a uniformly distributed~$x\in\pmset{k}$, we have
\begin{align*}
\Exp\Big[ \Big\|\sum_{i=1}^k x_i  A_i\Big\|_{S_\infty} \Big]
&\geq
\Exp\Big[\frac{D(x)^{\mathsf T}}{\sqrt{(N')^l}}\left( \sum_{i=1}^k x_i A_i\right) \frac{D(x)}{\sqrt{(N')^l}}\Big]\\
&\stackrel{\eqref{eq:ASSdef}}{=}
\Exp\Big[\frac{2}{(N')^l}\sum_{i=1}^k\sum_{\{{\bf S}, {\bf S'}\}\in\mathcal P_i} x_i\, D(x)_{\bf S}D(x)_{\bf S'}\, A_i({\bf S}, {\bf S'})  \Big]\\
&\stackrel{\eqref{eq:DSS}}{=}
\frac{2}{(N')^l}\sum_{i=1}^k
\sum_{\{{\bf S}, {\bf S'}\}\in\mathcal P_i} \Big|\Exp_x\Big[x_i\, \prod_{s\in S}C(x)_{s} \Big]\Big|
\stackrel{\eqref{eq:match-exp}}{\geq}
\frac{\eps k}{8^{r}}.
\end{align*}
\end{proof}

\bibliographystyle{alpha}
\bibliography{ldcl1}

%

\end{document}